\documentclass{article}
\usepackage{mathrsfs}
\usepackage{bbm}
\usepackage{amsfonts}
\usepackage{mathrsfs}
\usepackage{amsthm}
\usepackage{latexsym}
\usepackage{amsmath}

\newtheorem{theorem}{{Theorem}}

\newtheorem{lemma}{{Lemma}}
\newtheorem{remark}{{Remark}}

\newtheorem*{open problem}{{Open problem}}

\usepackage{lastpage}
\usepackage{epsfig}

\def\cdd{\mbox{\boldmath$\cdot$}~}
\makeatletter
\def\@oddfoot{\hfill}

\newcount\shumeicount
\def\setshumei#1#2#3{%
  \shumeicount=\count0
  \def\@oddhead{%
    \raise-5pt\hbox to0pt{\vrule width\hsize height 0pt depth 0.4pt\hss}\relax
    \ifnum \shumeicount=\count0
      \raise-7pt\hbox to0pt{\vrule width\hsize height 0pt depth 0.4pt\hss}\relax
      #1
    \else
      \ifodd\count0
        #2
      \else
        #3
       \fi
     \fi
  }%
}
\makeatother

\makeatletter
\def\@oddfoot{\hfill}
\newcount\shujiaocount
\def\setshujiao{%
  \shujiaocount=\count0
  \def\@oddfoot{%
      \ifodd\count0
      \else
      \fi
  }%
}
\makeatother
\def\dse#1{\vskip 0.6cm\noindent
        {\large\bf #1}
        \vskip 0.4cm}

\def\dse#1{\vskip 0.6cm\noindent
        {\large\bf #1}
        \vskip 0.4cm}

\def\biaoti#1#2#3#4{{
  \vspace*{0.3cm}
  \begin{flushleft} \Large\bf #1\end{flushleft}
  \vspace*{-0.2cm}
      \begin{flushleft}
      \bf #2
      \end{flushleft}
      \footnotetext{\hspace{-6mm} #3\\ #4}}}

\def\dshm#1#2#3#4
{\setshumei{ #2: {\thepage--\pageref{LastPage}}\hfill}
            {\hfill {\small #3}\hfill\hbox to0pt{\hss\thepage}}
            {\hbox to0pt{\thepage\hss}\hfill {\small #4}\hfill
            }
            \setshujiao}
\def\drd#1
{{\vskip 1cm\small \begin{flushleft}
 #1 \\
\copyright 2006 Springer Science + Business Media, Inc.
\end{flushleft}}}
\def\dab#1#2{\noindent {{\small\bf Abstract~~}}{{\small #1}}
            \vskip 0.1cm
             \noindent{{\small\bf Key words~~}}{{\small #2}}
                      }

\def\dse#1{\vskip 0.6cm\noindent
        {\large\bf #1}
        \vskip 0.4cm}

\def\Rmn#1{\expandafter\uppercase\expandafter{\romannumeral #1}}

\usepackage{amssymb}
\begin{document}

\biaoti{Binary duadic codes and their related codes with a square-root-like lower bound}
{Tingting Wu$^{1}$  \cdd Lanqiang Li$^{1}$  \cdd Xiuyu Zhang$^{2}$ \cdd Shixin Zhu$^{2}$ }
{*Corresponding author}
{Tingting Wu \\{ Email:
wutingting58@163.com\\
Lanqiang Li \\Email:lilanqiang716@126.com\\
Xiuyu Zhang \\ Email:xyzfuman@mail.hfut.edu.cn\\
Shixin Zhu \\ Email:zhushixinmath@hfut.edu.cn\\
\\{}}\\
{1. School of Information and Artificial Intelligence, Anhui Agricultural University, Hefei 230036, Anhui, China\\
2. School of Mathematics, Hefei University of Technology, Hefei 230009, Anhui, China 
}}

\dshm{}{T.T. Wu $\cdot$ L.Q. Li $\cdot$ X.Y. Zhang $\cdot$ S.X. Zhu }{T.T. Wu $\cdot$ L.Q. Li  $\cdot$ X.Y. Zhang $\cdot$ S.X. Zhu }

\dab{Binary cyclic codes have been a hot topic for many years, and significant progress has been made in the study of this types of codes. As is well known, it is hard to construct infinite families of binary cyclic codes $[n,\frac{n\pm 1}{2}]$ with good minimum distance. In this paper, by using the BCH bound on cyclic codes, one of the open problems proposed by Liu et al. about binary cyclic codes (Finite Field Appl 91:102270, 2023) is settled. Specially, we present several families of binary duadic codes with length $2^m-1$ and dimension $2^{m-1}$, and the minimum distances have a square-root-like lower bound. As a by-product, the parameters of their dual codes and extended codes are provided, where the latter are self-dual and doubly-even. }
{Cyclic code; Duadic code; BCH bound; Self-dual code}

\dse{1~~Introduction}

Let $p$ be a prime and $\mathbb{F}_{p^m}$ be a finite field with $p^m$ elements. A $p$-ary $[n, k, d]$ linear code $\mathcal{C}$ is a $k$-dimensional subspace of $\mathbb{F}_{p}^{n}$ with minimum distance $d$. Define the dual code $\mathcal{C}^{\perp}$ of $\mathcal{C}$ as 
$$
\mathcal{C}^{\perp}=\{\vec{u}\in \mathbb{F}_{p}^{n}| \vec{u}\cdot \vec{v}=0~\text{for any} ~\vec{v} \in \mathcal{C} \}.
$$
If $\mathcal{C} \bigcap \mathcal{C}^{\perp}=0$, then $C$ is called linear complementary dual(LCD); and if $\mathcal{C}=\mathcal{C}^{\perp}$, $\mathcal{C}$ is called self-dual. The linear code $\mathcal{C}$ is said to be a cyclic code if the cyclic shift of any codeword in $\mathcal{C}$ is still a codeword, i.e. $(a_{n-1}, a_0, a_1, \cdots, a_{n-2}) \in \mathcal{C}$ when $(a_0, a_1, \cdots, a_{n-1}) \in \mathcal{C}$. For convenience, we can regard any codeword $(a_0, a_1, \cdots, a_{n-1}) \in \mathcal{C}$ as a polynomial $a_0+a_1x+\cdots+a_{n-1}x^{n-1} \in \frac{\mathbb{F}_{p}[x]}{\langle x^n-1 \rangle}$. From \cite{1}, it is known that the cyclic code $\mathcal{C}$ of length $n$ over $\mathbb{F}_{p}$ is an ideal of the ring $\frac{\mathbb{F}_{p}[x]}{\langle x^n-1 \rangle}$. Since each ideal of $\frac{\mathbb{F}_{p}[x]}{\langle x^n-1 \rangle}$ is principal, then there is a monic polynomial $g(x)$ of lowest degree such that $\mathcal{C}=\langle g(x) \rangle$. The polynomial $g(x)$ is called the generator polynomial of $\mathcal{C}$. Additionally,  $h(x)=\frac{x^n-1}{g(x)}$ is called the check polynomial of the cyclic code $\mathcal{C}$. Note that the dual code $\mathcal{C}^{\perp}$ of $\mathcal{C}$ is generated by the reciprocal polynomial of $h(x)$.

 Cyclic codes have an extensive applications in communication systems, consumer electronics and data storage systems because of their efficient encoding and decoding algorithms\cite{2}-\cite{4}. Although it is easy to construct infinite families of binary cyclic codes with parameters $[n, \frac{n \pm 1}{2}]$ (where $n$ is odd), achieving good minimum distances for these codes is a significant challenge. There is a class of quadratic residue codes \cite{5} of length $n$ and dimension $\frac{n+1}{2}$, and the minimum distance $d$ have a square-root lower bound $\sqrt{n}$. As a generalization of quadratic residue codes, duadic codes were introduced and studied in \cite{6}-\cite{8}, where many properties were proved. Additionally, [8] provided all binary duadic codes of length up to 241. When the length of binary duadic codes $\mathcal{C}$ is prime power, the total number of these codes and their constructions were given in \cite{9,10}. 

Recently, Tang and Ding \cite{11} presented a family of binary duadic codes $\mathcal{C}$ with parameters $[2^m-1,2^{m-1},d]$, where $d \geq 2^{\frac{m-1}{2}}+1$ when $m \equiv{3} \pmod{4}$ and $d \geq 2^{\frac{m-1}{2}}+3$ when $m \equiv{1} \pmod{4}$. The polynomial
$$
g_{(i,m)}(x)=\prod_{j \in T_{(i,m)}} (x-\alpha^{j})
$$
is the generator polynomial of the code $\mathcal{C}$, where $\alpha$ is the generator of $\mathbb{F}_{2^m}^{*}$, $T_{(i,m)}=\{1\leq j \leq {2^m-2}: w_{2}(j) \equiv{i} \pmod{2}\}$ for $i\in\{0,1 \}$. Liu and Li et al. \cite{12} constructed two families of binary duadic codes of length $2^m-1$ and dimension $2^{m-1}$, its generator polynomial is 
$$
g_{(i_1, i_2, m)}(x)=\prod_{j \in T_{(i_1,i_2,m)}} (x-\alpha^{j}),
$$ 
where $i_1, i_2 \in \{0,1,2,3\}$ are distinct elements, $T_{(i_1,i_2, m)}=\{ 1\leq j \leq {2^m-2}:w_{2}(j) \equiv{i_1~\text{or}~i_2 } \pmod{4}  \}$. Moreover, they gave three families of binary cyclic codes of length $n=2^m-1$ and dimension $\frac{n-6}{3} \leq k \leq \frac{2(n+6)}{3}$, which have good lower bound on their minimum distance and contain distance-optimal codes. Finally, they put forward four open problems. Inspired by this work, Liu and Li et al. \cite{13} partially solved  open problem 4, which presented binary duadic codes and their related codes with generator polynomial
$$
g_{[6,m,S]}(x)=\prod_{j \in T_{[6,m,S]}}(x-\alpha^{j}),
$$
where $T_{[6,m,S]}=\{1\leq j \leq{2^m-2}:w_{2}(j) \pmod{6} \in S  \}$, $S\subsetneq \mathbb Z_{6}$ and $|S|=3$.

For lack of infinite family of ternary cyclic codes with parameters $[n, k\in \{\frac{n}{2}, \frac{(n \pm 2)}{2}\}, d \geq \sqrt{n}]$, Chen and Ding et al. \cite{14} constructed four such infinite families of ternary cyclic codes and their dual codes with a square-root-like lower bound. The generator polynomial is 
$$
g_{(i_1,i_2,m)}(x)=\prod_{
	\begin{array}{c}
	1\leq j \leq{3^m-2}\\
	 w_{3}(j) \pmod{4} \in\{i_1,i_2\}
 	\end{array}}(x-\beta^j),
$$
where $i_1, i_2\in \{0,1,2,3\}$ are distinct elements, $\beta$ is the generator of $\mathbb{F}_{3^m}^{*}$. 

Inspired by Tang-Ding codes, Shi and Tao et al. \cite{15} proposed an analogue for quaternary codes with generator polynomial 
$$
g_{(i,m)}(x)=\prod_{\begin{array}{c} 1\leq j \leq {4^m-2}\\ w_{2}(j) \equiv{i} \pmod{2} \end{array}} (x-\alpha^j),
$$
where $i \in \{0,1\}$. When $m$ is odd, they obtained a pair of odd-like duadic codes; when $m$ is even, they got a new family of LCD cyclic codes with rate close to $\frac{1}{2}$. After that, Sun and Li et al. \cite{16} studied $2^s$-ary Tang-Ding codes, where $s\geq 2$. They presented an infinite family of $2^s$-ary duadic codes with a square-root-like lower bound. For more cyclic codes with good parameters, reader can refer to Ref.\cite{17}-\cite{26}.  

\begin{open problem}\cite{12} 
Let $r \geq 6$ be an even integer. Find a subset $S$ of $\mathbb{Z}_{r}$ with $|S|=\frac{r}{2}$ such that $\mathcal{C}_{[r,m,S]}$ is a binary duadic code of length $n=2^m-1$ for infinitely many odd $m$. Determine the parameters of these duadic codes. 
\end{open problem}
	
In this paper, we will solve this problem. By using the BCH bound on cyclic codes, we give several families of binary duadic codes with parameters $[2^m-1, 2^{m-1}, d]$, where the lower bounds on their minimum distances $d$ are close to the square-root bound. In addition, we provide the parameters of the dual codes and extended codes of these binary duadic codes.

This paper is organized as follows. Section 2 gives some basic knowledge needed later. Section $3$ presents the construction of binary cyclic codes. Sections 4 and 5 construct several families of binary duadic codes and their dual codes and extended codes, which all have square-root-like lower bounds. Section 6 provides an example for $r=8$. Section 7 concludes this paper.

\dse{2~~Preliminaries}

In this section, we give some basic knowledge which will be employed later.

Let $n=2^m-1$ and $\mathbb{Z}_{n}=\{0,1,\cdots, n-1\}$. For any $s \in \mathbb{Z}_n$, the $2$-cyclotomic coset modulo $n$ is defined as 
$$
C_s=\{s\cdot 2^t \pmod{n}: 0\leq t \leq l_s-1\},
$$
where $l_s$ is the least positive integer such that $2^{l_s}\cdot s \equiv{s} \pmod{n}$. 

Let $\alpha$ be an $n$-th primitive root of unity in $\mathbb{F}_{2^m}$. Let $\mathcal{C}$ be a binary cyclic code of length $n$ with generator polynomial $g(x)$. Define the defining set of $\mathcal{C}$ as 
$$
T=\{0\leq i \leq {n-1}: g(\alpha^i)=0\}.
$$
Moreover, $T$ is the union of some $2$-cyclotomic cosets modulo $n$.

Let $S_1$ and $S_2$ be two subsets of $\mathbb{Z}_{n}$. Both $S_1$ and $S_2$ are the union of some 2-cyclotomic cosets modulo $n$. Let $\mathcal{C}_{1}$ and $\mathcal{C}_{2}$ be cyclic codes over $\mathbb{F}_{2}$ with defining sets $T_1=S_1$ and $T_2=S_2$, respectively. Then $\mathcal{C}_{1}$ and $\mathcal{C}_{2}$ form a pair of odd-like duadic codes if and only if 

$(1)$ $S_1$ and $S_2$ satisfy
$$
S_1 \bigcup S_2=\mathbb{Z}_{n}\setminus \{0\} ~\text{and}~ S_1 \bigcap S_2=\emptyset.
$$

$(2)$ There is a unit $\mu \in \mathbb{Z}_{n}$ such that
$$
S_1\cdot \mu=S_2 ~\text{and}~ S_2 \cdot \mu =S_1.
$$

From above, we know that $(S_1, S_2, \mu)$ is a splitting of $\mathbb{Z}_{n}$. Similarly, if $\widetilde{\mathcal{C}_1}$ and $\widetilde{\mathcal{C}_2}$ are cyclic codes over $\mathbb{F}_{2}$ with defining sets $\widetilde{T_1}=\{0\} \bigcup S_1$ and  
$\widetilde{T_2}=\{0\} \bigcup S_2$, then they form a pair of even-like duadic codes. By definition, $\mathcal{C}_{1}$, $\mathcal{C}_{2}$ have parameters $[n,\frac{n+1}{2}]$ and $\widetilde{\mathcal{C}_1}$, $\widetilde{\mathcal{C}_2}$ have parameters $[n,\frac{n-1}{2}]$.

For odd-like duadic codes, we have the following result \cite[Theorem 6.5.2]{5}.

\begin{theorem}(Square root bound)\label{theorem 1}
Let $\mathcal{C}_{1}$ and $\mathcal{C}_{2}$ be a pair of odd-like duadic codes of length $n$ over $\mathbb{F}_{2}$. Let $d_0$ be their (common) minimum odd weight. Then the following hold:

(1) $d_{0}^{2}\geq n$.

(2) If the splitting defining the duadic codes is given by $\mu=-1$, then $d_{0}^2-d_0+1 \geq n$.

(3) Suppose $d_{0}^2-d_0+1=n$, where $d_0 >2$, and assume that the splitting defining the duadic codes is given by $\mu=-1$. Then $d_0$ is the minimum weight of both $\mathcal{C}_{1}$ and $\mathcal{C}_{2}$.
		
\end{theorem}

Now we consider extending odd-like duadic codes over $\mathbb{F}_{2}$. Define the extended code $\overline{\mathcal{C}}$ of $\mathcal{C}$ as 
$$
\overline{\mathcal{C}}=\{ (c_0, c_1, \cdots, c_{n-1}, c_n): (c_0, c_1, \cdots, c_{n-1})\in \mathcal{C}~\text{with}~ \sum_{i=0}^{n}c_i=0  \}.
$$

From \cite[Theorem 6.4.12 and Theorem 6.5.1]{5}, we have following results:

\begin{theorem}\label{theorem 2}
Let $\mathcal{C}_{1}$ and $\mathcal{C}_{2}$ be odd-like duadic codes of length $n$ over $\mathbb{F}_{2}$. Denote $\overline{\mathcal{C}_{1}}$ and $\overline{\mathcal{C}_{2}}$ are extended codes of $\mathcal{C}_{1}$ and $\mathcal{C}_{2}$, respectively. If $\mu=-1$ gives the splitting for  $\mathcal{C}_{1}$ and $\mathcal{C}_{2}$, then $\overline{\mathcal{C}_{1}}$ and $\overline{\mathcal{C}_{2}}$ are self-dual.  
\end{theorem}

\begin{theorem}\label{theorem 3}
Let $\mathcal{C}_{1}$ and $\mathcal{C}_{2}$ be odd-like duadic codes of length $n$ with splitting given by $\mu=-1$. Then $\overline{\mathcal{C}_{1}}$ and $\overline{\mathcal{C}_{2}}$ are self-dual doubly-even if $n \equiv{-1} \pmod{8}$ or self-dual singly-even if $n \equiv{1} \pmod{8}$. 

\end{theorem} 

The following lemmas will be used later.

\begin{lemma}\label{lemma 1}
	Let $a \geq 2$ and $l, m$ be positive integers, then 
	$$
	\gcd(a^m-1, a^l-1)=a^{\gcd(m,l)}-1.
	$$
\end{lemma}

\begin{lemma}\label{lemma 2}\cite{1} (BCH bound)
	Let $g(x)\in \mathbb{F}_{q}[x]$ be the generator polynomial of a cyclic code $\mathcal{C}$ with length $n$. If $g(\alpha^{l+ic})$=0 with respect to $n$-th primitive root of unity $\alpha$, where $\gcd(n,c)=1$ and $0\leq i \leq {d_0-2}$, then $d\geq d_0$.  
	
\end{lemma}

\dse{3~~The construction of binary duadic   codes}

For any integer $0 \leq i \leq {2^m-1}$, the 2-adic expansion of $i$ is 
$$
i=i_0+i_1\cdot 2+i_2\cdot 2^2+\cdots+i_{m-1}\cdot 2^{m-1},
$$
where $i_j\in \{0,1\}$ for $0\leq j \leq {m-1}$. And define $w_{2}(j)=\sum\limits_{j=0}^{m-1}i_j$.

Let $r$ be even and $n=2^m-1$ for odd $m \geq 3$. Let $S\subsetneq \mathbb{Z}_{r}$ and $|S|=\frac{r}{2}$. Define
$$
T_{[r,m,S]}=\{1\leq j \leq {n-1}: w_{2}(j)\pmod{r} \in S \}.
$$
Note that $T_{[r,m,S]}$ is the union of some 2-cyclotomic cosets modulo $n$.

Let $\mathbb{F}_{2^m}^{*}=\langle \alpha \rangle$. Denote $\mathcal{C}_{[r,m,S]}$ is the binary cyclic code of length $n$ with generator polynomial 
$$
g_{[r,m,S]}(x)=\prod_{i \in T_{[r,m,S]}} (x-\alpha^i).
$$

In this paper, we denote $S=\{ i_1,i_2,\cdots, i_{\frac{r}{2}}\}$, where $i_1, i_2, \cdots, i_{\frac{r}{2}} \in \mathbb{Z}_{r}$ are distinct elements. As $m$ is odd, we have $m \equiv{t} \pmod{r}$ for odd $t$. Furthermore, let $S'=\mathbb{Z}_{r}\setminus{S}=\{i_1',i_2',\cdots, i_{\frac{r}{2}}'\}$ and $i_{j}+i_{j}'\equiv{t} \pmod{r}$ for $1\leq j \leq {\frac{r}{2}}$, where $i_{1}, i_2, \cdots, i_{\frac{r}{2}}, i_{1}', i_2', \cdots, i_{\frac{r}{2}}' $ are pairwise distinct. We will prove that $\mathcal{C}_{[r,m,S]}$ and $\mathcal{C}_{[r,m,S']}$ form a pair of odd-like duadic codes under certain conditions, where $r$ is even and $m \geq 3$ is odd. Moreover, the dual codes and extended codes of these duadic codes are investigated.

\dse{4~~Some auxiliary results}

To give lower bounds on the minimum distances of the above duadic codes, we need some auxiliary results in the following. 

\begin{lemma}\label{lemma 3}
Let $m \equiv{t} \pmod{r}$ with odd $t\neq 3$ and even $r>2$. When $S$ contains elements $\frac{t-1}{2}, \frac{t+r-1}{2}$ and $t-1$, $S'$ contains $\frac{t+1}{2}$ and $\frac{t+r+1}{2}$, we have

(1) If $v=2^{\frac{m-1}{2}}-1$, then $\gcd(v,n)=1$ and 
$$
\{av: 1\leq a \leq 2^{\frac{m-1}{2}}+2\}\subseteq T_{[r,m,S]}.
$$

(2) If $v=2^{\frac{m+1}{2}}-1$, then $\gcd(v,n)=1$ and
$$
\{av: 1\leq a \leq 2^{\frac{m-1}{2}}+2\}\subseteq T_{[r,m,S']}.
$$
	
\end{lemma}

\begin{proof}
(1) If $v=2^{\frac{m-1}{2}}-1$, then $\gcd(v,n)=2^{\gcd(\frac{m-1}{2},m)}-1=1$ by Lemma \ref{lemma 1}. When $a=2^{\frac{m-1}{2}}+2$, we have
\[
av=(2^{\frac{m-1}{2}}+2)(2^{\frac{m-1}{2}}-1)=2^{m-1}+2^{\frac{m-1}{2}}-2=2(2^{m-2}+2^{\frac{m-3}{2}}-1).
\]
Then, $w_{2}(av)=\frac{m-1}{2}$. Due to $m \equiv{t} \pmod{r}$, $m \equiv{t} \pmod{2r}$ or $m \equiv{t+r} \pmod{2r}$. Thus, $w_{2}(av) \pmod{r} ={\frac{t-1}{2}}$ or ${\frac{t+r-1}{2}}$ $\in S$. When $a=2^{\frac{m-1}{2}}+1$, $av=(2^{\frac{m-1}{2}}+1)(2^{\frac{m-1}{2}}-1)=2^{m-1}-1$. Then, $w_{2}(av)=m-1 \equiv{t-1} \pmod{r} \in S$. When $a=2^{\frac{m-1}{2}}$, we have $w_{2}(av)=w_{2}(v)=\frac{m-1}{2}$ $\pmod{r} ={\frac{t-1}{2}}$ or ${\frac{t+r-1}{2}}$ $\in S$. 

Now we consider that $1\leq a \leq {2^{\frac{m-1}{2}}-1}$. Let $a=2^l\cdot\overline{a}$ with odd $\overline{a}$, where $l\geq 0$ is an integer. Then we have $1\leq \overline{a} \leq {2^{\frac{m-1}{2}}-1}$ and denote $\overline{a}$ as
$$
\overline{a}=\sum_{i=0}^{\frac{m-3}{2}}a_i\cdot 2^i, ~~a_i\in\{0,1\}.
$$ 
Since $\overline{a}$ is odd, $a_0=1$. It is known that
$$
\overline{a}v=(\sum_{i=0}^{\frac{m-3}{2}}a_i\cdot 2^i)(2^{\frac{m-1}{2}}-1)=\sum_{i=1}^{\frac{m-3}{2}}a_{i}2^{i+\frac{m-1}{2}}+\sum_{i=0}^{\frac{m-3}{2}}(1-a_i)2^i+1.
$$
So, $w_{2}(av)=w_{2}(\overline{a}v)=w_{2}(\overline{a})-1+\frac{m-1}{2}-w_{2}(\overline{a})+1=\frac{m-1}{2} \pmod{r}= \frac{t-1}{2}$ or $\frac{t+r-1}{2} \in S$. 

As a result, we have $\{av: 1\leq a \leq 2^{\frac{m-1}{2}}+2\}\subseteq T_{[r,m,S]}$.

(2) If $v=2^{\frac{m+1}{2}}-1$, it is easy to deduce that $\gcd(v,n)=1$. When $a=2^{\frac{m-1}{2}}+2$, $av=(2^{\frac{m-1}{2}}+2)(2^{\frac{m+1}{2}}-1)=2(2^{m-1}+2^{\frac{m-1}{2}}+2^{\frac{m-3}{2}}-1)$. Then we have $w_{2}(av)=\frac{m+1}{2} \pmod{r}=\frac{t+1}{2}$ or $\frac{t+r+1}{2}$$\in S'$. When $a=2^{\frac{m-1}{2}}+1$, $av=(2^{\frac{m-1}{2}}+1)(2^{\frac{m+1}{2}}-1)=2^m+2^{\frac{m-1}{2}}-1$. Then, $w_{2}(av)=\frac{m+1}{2} \pmod{r}=\frac{t+1}{2}$ or $\frac{t+r+1}{2} \in S'$. When $a=2^{\frac{m-1}{2}}$, we have $w_{2}(av)=w_2(v)=\frac{m+1}{2} \in S'$. Next, we consider that $1\leq a \leq {2^{\frac{m-1}{2}}-1}$. We denote $a=2^l\cdot \overline{a}$, where $l\geq 0$ is an integer and $\overline{a}$ is odd. Then $1\leq \overline{a} \leq{2^{\frac{m-1}{2}}-1}$ and denote $\overline{a}$ as $\overline{a}=\sum_{i=0}^{\frac{m-3}{2}}a_i\cdot 2^i$, where $a_i\in\{0,1\}$. By calculating
$$
\overline{a}v=(\sum_{i=0}^{\frac{m-3}{2}}a_i\cdot 2^i)(2^{\frac{m+1}{2}}-1)=\sum_{i=1}^{\frac{m-3}{2}}a_{i}2^{i+\frac{m+1}{2}}+\sum_{i=0}^{\frac{m-3}{2}}(1-a_i)2^i+2^{\frac{m-1}{2}}+1,
$$
we can get $w_{2}(av)=w_{2}(\overline{a}v)=w_{2}(\overline{a})-1+\frac{m-1}{2}-w_{2}(\overline{a})+2=\frac{m+1}{2} \pmod{r}=\frac{t+1}{2}$ or $\frac{t+r+1}{2} \in S'$.

From above, we have $\{av: 1\leq a \leq 2^{\frac{m-1}{2}}+2\}\subseteq T_{[r,m,S']}$.
\qedhere
\end{proof}

Similar to the proof of Lemma \ref{lemma 3}, we have the following three lemmas.

\begin{lemma} \label{lemma 4}
Let $m \equiv{t} \pmod{r}$ with odd $t\neq 1$. When $S$ contains elements $\frac{t-1}{2}, \frac{t+r-1}{2}$ and $1$, $S'$ contains $\frac{t+1}{2}$ and $\frac{t+r+1}{2}$, we have

(1) If $v=2^{\frac{m-1}{2}}-1$, then $\gcd(v,n)=1$ and 
$$
\{av: 1\leq a \leq 2^{\frac{m-1}{2}}\} \subseteq T_{[r,m,S]}.
$$ 

(2) If $v=2^{\frac{m+1}{2}}-1$, then $\gcd(v,n)=1$ and 
$$
\{av: 1\leq a \leq 2^{\frac{m-1}{2}}\} \subseteq T_{[r,m,S']}.
$$ 

\end{lemma}

\begin{lemma} \label{lemma 5}
Let $m \equiv{t} \pmod{r}$ with odd $t\neq 3$. When $S$ contains elements $\frac{t-1}{2}, \frac{t+r+1}{2}$ and $t-1$, $S'$ contains $\frac{t+1}{2}$ and $\frac{t+r-1}{2}$, we have

(1) When $m \equiv{t} \pmod{2r}$,

(1.1) If $v=2^{\frac{m-1}{2}}-1$, then $\gcd(v,n)=1$ and 
$$
\{av: 1\leq a \leq 2^{\frac{m-1}{2}}+2\}\subseteq T_{[r,m,S]}.
$$

(1.2) If $v=2^{\frac{m+1}{2}}-1$, then $\gcd(v,n)=1$ and
$$
\{av: 1\leq a \leq 2^{\frac{m-1}{2}}+2\}\subseteq T_{[r,m,S']}.
$$

(2) When $m \equiv{t+r} \pmod{2r}$,

(2.1) If $v=2^{\frac{m+1}{2}}-1$, then $\gcd(v,n)=1$ and 
$$
\{av: 1\leq a \leq 2^{\frac{m-1}{2}}\}\subseteq T_{[r,m,S]}.
$$

(2.2) If $v=2^{\frac{m-1}{2}}-1$, then $\gcd(v,n)=1$ and
$$
\{av: 1\leq a \leq 2^{\frac{m-1}{2}}\}\subseteq T_{[r,m,S']}.
$$
\end{lemma}

\begin{lemma} \label{lemma 6}
	Let $m \equiv{t} \pmod{r}$ with odd $t\neq 1$. When $S$ contains elements $\frac{t-1}{2}, \frac{t+r+1}{2}$ and $1$, $S'$ contains $\frac{t+1}{2}$ and $\frac{t+r-1}{2}$, we have
	
	(1) When $m \equiv{t} \pmod{2r}$,
	
	(1.1) If $v=2^{\frac{m-1}{2}}-1$, then $\gcd(v,n)=1$ and 
	$$
	\{av: 1\leq a \leq 2^{\frac{m-1}{2}}\}\subseteq T_{[r,m,S]}.
	$$
	
	(1.2) If $v=2^{\frac{m+1}{2}}-1$, then $\gcd(v,n)=1$ and
	$$
	\{av: 1\leq a \leq 2^{\frac{m-1}{2}}\}\subseteq T_{[r,m,S']}.
	$$
	
	(2) When $m \equiv{t+r} \pmod{2r}$,
	
	(2.1) If $v=2^{\frac{m+1}{2}}-1$, then $\gcd(v,n)=1$ and 
	$$
	\{av: 1\leq a \leq 2^{\frac{m-1}{2}}+2\}\subseteq T_{[r,m,S]}.
	$$
	
	(2.2) If $v=2^{\frac{m-1}{2}}-1$, then $\gcd(v,n)=1$ and
	$$
	\{av: 1\leq a \leq 2^{\frac{m-1}{2}}+2\}\subseteq T_{[r,m,S']}.
	$$
	
\end{lemma}

\begin{remark}
Note that the elements in $S$ are all modulo $r$.
\end{remark}

\dse{5~~Parameters of the codes $\mathcal{C}_{[r,m,S]}$ and their related codes}

In this section, we study the parameters of the codes $\mathcal{C}_{[r,m,S]}$ and their related codes, where $m \geq 3$. First, we consider the case that $m\equiv{t} \pmod{r}$ for odd $t\neq 3$ and even $r>2$, the set $S$ contains elements $\frac{t-1}{2}$, $\frac{t+r-1}{2}$ and $t-1$ or $\frac{t+1}{2}$, $\frac{t+r+1}{2}$ and $1$, the parameters of $\mathcal{C}_{[r,m,S]}$ and $\mathcal{C}_{[r,m,S']}$ are given as follows.

\begin{theorem}\label{theorem 4}
Let $m\equiv{t} \pmod{r}$ with odd $t\neq 3$ and even $r>2$. When $S$ contains elements $\frac{t-1}{2}$, $\frac{t+r-1}{2}$ and $t-1$ or $\frac{t+1}{2}$, $\frac{t+r+1}{2}$ and 1, the codes $\mathcal{C}_{[r,m,S]}$ and $\mathcal{C}_{[r,m,S']}$ form a pair of odd-like duadic codes with parameters $[2^m-1,2^{m-1},d\geq 2^{\frac{m-1}{2}}+3]$.

\end{theorem}
	
\begin{proof}
From the previous discussions, we know that $S=\{ i_1,i_2,\cdots, i_{\frac{r}{2}}\}$ and $S'=\{i_1', i_2', \cdots, i_{\frac{r}{2}}' \}$, where $i_j+i_{j}'\equiv{t} \pmod{r}$ for $1\leq j \leq \frac{r}{2}$. Since $w_{2}(i)
=m-w_{2}(n-i)$ for $1\leq i \leq {n-1}$, then $i \in T_{[r,m,S]}$ if and only if $n-i \in T_{[r,m,S']}$. Then $T_{[r,m,S]}$ and $T_{[r,m,S']}$ partition $\mathbb{Z}_{n}\setminus\{0\}$ and 
$$
T_{[r,m,S]}=-T_{[r,m,S']}~\text{and}~T_{[r,m,S']}=-T_{[r,m,S]}.
$$

Therefore, the codes $\mathcal{C}_{[r,m,S]}$ and $\mathcal{C}_{[r,m,S']}$ form a pair of odd-like duadic codes with length $n=2^m-1$ and dimension $\frac{n+1}{2}=2^{m-1}$. Moreover, they have the same minimum distance $d$.

Next, we need to prove the lower bounds on minimum distance $d$. If $v=2^{\frac{m-1}{2}}-1$, it is easy to deduce that $\gcd(v,n)=1$ by Lemma \ref{lemma 1}. Let $\overline{v}$ be the integer that satisfies $v\overline{v}\equiv{1} \pmod{n}$. Define $\gamma=\alpha^{\overline{v}}$. From Lemma \ref{lemma 3}, it is known that the defining set of $\mathcal{C}_{[r,m,S]}$ relative to $\gamma$ includes the set $\{1,2,\cdots,2^{\frac{m-1}{2}}+2\}$. Hence, we can get the minimum distance $d\geq  2^{\frac{m-1}{2}}+3$ from the BCH bound on cyclic codes. If $v=2^{\frac{m+1}{2}}-1$, we also can obtain the same results. 
\qedhere
\end{proof}

Under the above conditions, the parameters of the dual codes $\mathcal{C}_{[r,m,S]}^{\perp}$ and $\mathcal{C}_{[r,m,S']}^{\perp}$ are provided in the following.

\begin{theorem}\label{theorem5}
Let $m\equiv{t} \pmod{r}$ with odd $t\neq 3$ and even $r>2$. When $S$ contains elements $\frac{t-1}{2}$, $\frac{t+r-1}{2}$ and $t-1$ or $\frac{t+1}{2}$, $\frac{t+r+1}{2}$ and 1, the dual codes $\mathcal{C}_{[r,m,S]}^{\perp}$ and $\mathcal{C}_{[r,m,S']}^{\perp}$ form a pair of even-like duadic codes with parameters $[2^m-1,2^{m-1}-1,d^{\perp}\geq 2^{\frac{m-1}{2}}+4]$.
\end{theorem}

\begin{proof}
It is easy to know that the defining sets of $\mathcal{C}_{[r,m,S]}^{\perp}$ and $\mathcal{C}_{[r,m,S']}^{\perp}$ relative to $\alpha$ are $\{0\}\bigcup T_{[r,m,S]}$ and $\{0\}\bigcup T_{[r,m,S']}$, respectively. Then,  $\mathcal{C}_{[r,m,S]}^{\perp}$ and $\mathcal{C}_{[r,m,S']}^{\perp}$ are the even-weight subcodes of $\mathcal{C}_{[r,m,S]}$ and $\mathcal{C}_{[r,m,S']}$, respectively. Hence, we can get the desired conclusion.

\qedhere
\end{proof}

Next, the parameters of extended codes $\overline{\mathcal{C}_{[r,m,S]}}$ and 
$\overline{\mathcal{C}_{[r,m,S']}}$ are given.

\begin{theorem}\label{theorem 6}
Let $m\equiv{t} \pmod{r}$ with odd $t\neq 3$ and even $r>2$. When $S$ contains elements $\frac{t-1}{2}$, $\frac{t+r-1}{2}$ and $t-1$ or $\frac{t+1}{2}$, $\frac{t+r+1}{2}$ and 1, the extended codes $\overline{\mathcal{C}_{[r,m,S]}}$ and $\overline{\mathcal{C}_{[r,m,S']}}$ have parameters $[2^m,2^{m-1},\overline{d}\geq 2^{\frac{m-1}{2}}+4]$. Moreover, they are self-dual and doubly-even.

\end{theorem}

\begin{proof}
By Theorem \ref{theorem 4}, we know that $(T_{[r,m,S]}, T_{[r,m,S']}, -1)$ is a 
splitting of $\mathbb{Z}_{2^m-1}$ and the codes $\mathcal{C}_{[r,m,S]}$ and $\mathcal{C}_{[r,m,S']}$ form a pair of odd-like duadic codes. Then the extended codes $\overline{\mathcal{C}_{[r,m,S]}}$ and $\overline{\mathcal{C}_{[r,m,S']}}$ of $\mathcal{C}_{[r,m,S]}$ and $\mathcal{C}_{[r,m,S']}$ are self-dual from Theorem \ref{theorem 2}. Since $m \geq 3$, then $n=2^m-1\equiv{-1}\pmod{8}$. By Theorem \ref{theorem 3}, it is known that the extended codes $\overline{\mathcal{C}_{[r,m,S]}}$ and $\overline{\mathcal{C}_{[r,m,S']}}$ are doubly-even. The remaining proof is same as Theorem \ref{theorem 4}.  
\end{proof}

When the set $S$ satisfies the other conditions, using the similar method as Theorems \ref{theorem 4}-\ref{theorem 6} and from Lemmas \ref{lemma 4}-\ref{lemma 6}, we can obtain the parameters of the codes $\mathcal{C}_{[r,m,S]}$, $\mathcal{C}_{[r,m,S']}$ and their dual codes as well as the extended codes.

\begin{theorem}\label{theorem 7}
Let $m \equiv{t} \pmod{r}$ with odd $t\neq 1$. When $S$ contains elements $\frac{t-1}{2}$, $\frac{t+r-1}{2}$ and $1$, or $\frac{t+1}{2}$, $\frac{t+r+1}{2}$ and $t-1$, we have

(1) The codes $\mathcal{C}_{[r,m,S]}$ and $\mathcal{C}_{[r,m,S']}$ form a pair of odd-like duadic codes with parameters $[2^m-1, 2^{m-1},d\geq 2^{\frac{m-1}{2}}+1]$. 

(2) The dual codes $\mathcal{C}_{[r,m,S]}^{\perp}$ and $\mathcal{C}_{[r,m,S']}^{\perp}$ form a pair of even-like duadic codes with parameters $[2^m-1, 2^{m-1}-1,d^{\perp}\geq 2^{\frac{m-1}{2}}+2]$. 

(3) The extended codes $\overline {\mathcal{C}_{[r,m,S]}}$ and $\overline{\mathcal{C}_{[r,m,S']}}$ have parameters $[2^m, 2^{m-1},\overline{d}\geq 2^{\frac{m-1}{2}}+4]$. Moreover, they are self-dual and doubly-even.

\end{theorem}

\begin{theorem}\label{theorem 8}
	
	Let $m \equiv{t} \pmod{r}$ with odd $t\neq 3$. When $S$ contains elements $\frac{t-1}{2}$, $\frac{t+r+1}{2}$ and $t-1$, or $\frac{t+1}{2}$, $\frac{t+r-1}{2}$ and $1$, we have

(1) The codes $\mathcal{C}_{[r,m,S]}$ and $\mathcal{C}_{[r,m,S']}$ form a pair of odd-like duadic codes with parameters $[2^m-1, 2^{m-1},d]$, where

$$
d \geq \left\{
\begin{aligned}
&2^{\frac{m-1}{2}}+3, ~~~~&\mbox{if}&~~ m \equiv{t} \pmod{2r},\\
&2^{\frac{m-1}{2}}+1, ~~~~~&\mbox{if}&~~ m \equiv{t+r} \pmod{2r}.\\
\end{aligned}
\right.
$$

(2) The dual codes $\mathcal{C}_{[r,m,S]}^{\perp}$ and $\mathcal{C}_{[r,m,S']}^{\perp}$ form a pair of even-like duadic codes with parameters $[2^m-1, 2^{m-1}-1,d^{\perp}]$, where

$$
d^{\perp} \geq \left\{
\begin{aligned}
&2^{\frac{m-1}{2}}+4, ~~~~&\mbox{if}&~~ m \equiv{t} \pmod{2r},\\
&2^{\frac{m-1}{2}}+2, ~~~~~&\mbox{if}&~~ m \equiv{t+r} \pmod{2r}.\\
\end{aligned}
\right.
$$

(3) The extended codes $\overline {\mathcal{C}_{[r,m,S]}}$ and $\overline{\mathcal{C}_{[r,m,S']}}$ have parameters $[2^m, 2^{m-1},\overline{d}\geq 2^{\frac{m-1}{2}}+4]$. Moreover, they are self-dual and doubly-even.

\end{theorem}

\begin{theorem}\label{theorem 9}
	
	Let $m \equiv{t} \pmod{r}$ with odd $t\neq 1$. When $S$ contains elements $\frac{t-1}{2}$, $\frac{t+r+1}{2}$ and $1$, or $\frac{t+1}{2}$, $\frac{t+r-1}{2}$ and $t-1$, we have
	
	(1) The codes $\mathcal{C}_{[r,m,S]}$ and $\mathcal{C}_{[r,m,S']}$ form a pair of odd-like duadic codes with parameters $[2^m-1, 2^{m-1},d]$, where
	
	$$
	d \geq \left\{
	\begin{aligned}
	&2^{\frac{m-1}{2}}+1, ~~~~&\mbox{if}&~~ m \equiv{t} \pmod{2r},\\
	&2^{\frac{m-1}{2}}+3, ~~~~~&\mbox{if}&~~ m \equiv{t+r} \pmod{2r}.\\
	\end{aligned}
	\right.
	$$
	
	(2) The dual codes $\mathcal{C}_{[r,m,S]}^{\perp}$ and $\mathcal{C}_{[r,m,S']}^{\perp}$ form a pair of even-like duadic codes with parameters $[2^m-1, 2^{m-1}-1,d^{\perp}]$, where
	
	$$
	d^{\perp} \geq \left\{
	\begin{aligned}
	&2^{\frac{m-1}{2}}+2, ~~~~&\mbox{if}&~~ m \equiv{t} \pmod{2r},\\
	&2^{\frac{m-1}{2}}+4, ~~~~~&\mbox{if}&~~ m \equiv{t+r} \pmod{2r}.\\
	\end{aligned}
	\right.
	$$
	
	(3) The extended codes $\overline {\mathcal{C}_{[r,m,S]}}$ and $\overline{\mathcal{C}_{[r,m,S']}}$ have parameters $[2^m, 2^{m-1},\overline{d}\geq 2^{\frac{m-1}{2}}+4]$. Moreover, they are self-dual and doubly-even.
	
\end{theorem}

\begin{remark}
From above, we know that the parameters of the codes $\mathcal{C}_{[r,m,S]}$ and their related codes for $r=2$ and $|S|=1$$^{\cite{11}}$, or $r=4$ and $|S|=2$$^{\cite{12}}$, or $r=6$ and $|S|=3$$^{\cite{13}}$, can be determined by Theorems \ref{theorem 4}-\ref{theorem 9}.
\end{remark}

\dse{6~~Example: Constructions of all binary duadic codes for $r=8$}

To clearly present the conclusions drawn from Sect.5, we will provide an example in this section. Take $r=8$. Let $n=2^m-1$, $S \subsetneq \mathbb{Z}_8$ with $|S|=4$ and $S'= \mathbb{Z}_{8}\setminus S$. Denote
\[
T_{[8,m,S]}=\{1\leq j \leq{n-1}: w_{2}(j)\pmod{8} \in S \}
\]
is the defining set of the code $\mathcal{C}_{[8,,m,S]}$.

It is known that $|\{ S \subsetneq \mathbb{Z}_{8}: |S|=4\}|=\left(\begin{array}{c}8\\4 \end{array}\right)=70$. In the following, we present all binary duadic codes for $r=8$.

(1) When $m \equiv{1} \pmod{8}$, take $S\in\{\{0,2,3,4\}, \{0,2,3,5\}, \{0,2,4,6\}, \{0,2,5,6\}, \{0,3,4,7\},\\
 \{0,3,5,7\}, \{0,4,6,7\}, \{0,5,6,7\}\}$.
 
(2) When $m \equiv{3} \pmod{8}$, take $S\in\{\{0,1,4,5\}, \{0,1,4,6\}, \{0,1,5,7\}, \{0,1,6,7\}, \{0,2,4,5\}, \\
\{0,2,5,7\}, \{0,2,6,7\}, \{0,2,4,6\}\}$.

(3) When $m \equiv{5} \pmod{8}$, take $S\in\{\{0,1,2,6\}, \{0,1,2,7\}, \{0,1,3,6\}, \{0,1,3,7\}, \{0,2,4,7\}, \\
\{0,3,4,6\}, \{0,2,4,6\}, \{0,3,4,7\}\}$.

(4) When $m \equiv{7} \pmod{8}$, take $S\in\{\{0,1,2,3\}, \{0,1,2,4\}, \{0,1,3,5\}, \{0,2,3,6\}, \{0,3,5,6\},\\
 \{0,4,5,6\}, \{0,1,4,5\}, \{0,2,4,6\}\}$.

Note that the parameters of the codes $\mathcal C_{[8,m,S]}$ and $\mathcal{C}_{[8,m,S']}$ for $S=\{0,2,4,6\}$ had been investigated in \cite{11}; and $S\in\{\ \{0,3,4,7\}, \{0,1,4,5\} \}$ had been studied in \cite{12}. Hence, we focus on the parameters of $\mathcal C_{[8,m,S]}$ and $\mathcal{C}_{[8,m,S']}$ except $S\in\{\{0,2,4,6\}, \{0,3,4,7\}, \{0,1,4,5\} \}$.

\begin{theorem}
Let $m \equiv{1} \pmod{8}\geq 9$ be an integer.

(1) When $S=\{0,2,3,4\}$ or $\{0,4,6,7\}$, we have

(1.1) The codes $\mathcal C_{[8,m,S]}$ and $\mathcal{C}_{[8,m,S']}$ form a pair of odd-like duadic codes with parameters $[2^m-1, 2^{m-1}, d\geq 2^{\frac{m-1}{2}}+3]$.

(1.2) The dual codes $\mathcal C_{[8,m,S]}^{\perp}$ and $\mathcal{C}_{[8,m,S']}^{\perp}$ form a pair of even-like duadic codes and have parameters $[2^m-1,2^{m-1}-1,d^{\perp}\geq 2^{\frac{m-1}{2}}+4]$.

(1.3) The extended codes $\overline{\mathcal{C}_{[8,m,S]}}$ and $\overline{\mathcal{C}_{[8,m,S']}}$ have parameters $[2^m,2^{m-1},\overline{d}\geq 2^{\frac{m-1}{2}}+4]$. And  they are self-dual and doubly-even.
 
(2) When $S=\{0,2,3,5\}$ or $\{0,2,5,6\}$ or $\{0,3,5,7\}$ or $\{0,5,6,7\}$, we have

(2.1) The codes $\mathcal C_{[8,m,S]}$ and $\mathcal{C}_{[8,m,S']}$ form a pair of odd-like duadic codes with parameters $[2^m-1, 2^{m-1}, d]$, where 
$$
d \geq \left\{
\begin{aligned}
&2^{\frac{m-1}{2}}+1, ~~~~&\mbox{if}&~~ m \equiv{9} \pmod{16},\\
&2^{\frac{m-1}{2}}+3, ~~~~~&\mbox{if}&~~ m \equiv{1} \pmod{16}.\\
\end{aligned}
\right.
$$

(2.2) The dual codes $\mathcal C_{[8,m,S]}^{\perp}$ and $\mathcal{C}_{[8,m,S']}^{\perp}$ form a pair of even-like duadic codes and have parameters $[2^m-1,2^{m-1}-1,d^{\perp}]$, where 
$$
d^{\perp} \geq \left\{
\begin{aligned}
&2^{\frac{m-1}{2}}+2, ~~~~&\mbox{if}&~~ m \equiv{9} \pmod{16},\\
&2^{\frac{m-1}{2}}+4, ~~~~~&\mbox{if}&~~ m \equiv{1} \pmod{16}.\\
\end{aligned}
\right.
$$

(2.3) The extended codes $\overline{\mathcal{C}_{[8,m,S]}}$ and $\overline{\mathcal{C}_{[8,m,S']}}$ have parameters $[2^m,2^{m-1},\overline{d}\geq 2^{\frac{m-1}{2}}+4]$. And  they are self-dual and doubly-even.

\end{theorem}

\begin{theorem}
	Let $m \equiv{3} \pmod{8}$ be an integer.

(1) When $S=\{0,1,4,6\}$ or $\{0,1,6,7\}$ or $\{0,2,4,5\}$ or $\{0,2,5,7\}$, we have

(1.1) The codes $\mathcal C_{[8,m,S]}$ and $\mathcal{C}_{[8,m,S']}$ form a pair of odd-like duadic codes with parameters $[2^m-1, 2^{m-1}, d]$, where 
$$
d \geq \left\{
\begin{aligned}
&2^{\frac{m-1}{2}}+1, ~~~~&\mbox{if}&~~ m \equiv{3} \pmod{16},\\
&2^{\frac{m-1}{2}}+3, ~~~~~&\mbox{if}&~~ m \equiv{11} \pmod{16}.\\
\end{aligned}
\right.
$$

(1.2) The dual codes $\mathcal C_{[8,m,S]}^{\perp}$ and $\mathcal{C}_{[8,m,S']}^{\perp}$ form a pair of even-like duadic codes and have parameters $[2^m-1,2^{m-1}-1,d^{\perp}]$, where 
$$
d^{\perp} \geq \left\{
\begin{aligned}
&2^{\frac{m-1}{2}}+2, ~~~~&\mbox{if}&~~ m \equiv{3} \pmod{16},\\
&2^{\frac{m-1}{2}}+4, ~~~~~&\mbox{if}&~~ m \equiv{11} \pmod{16}.\\
\end{aligned}
\right.
$$

(1.3) The extended codes $\overline{\mathcal{C}_{[8,m,S]}}$ and $\overline{\mathcal{C}_{[8,m,S']}}$ have parameters $[2^m,2^{m-1},\overline{d}\geq 2^{\frac{m-1}{2}}+4]$. And  they are self-dual and doubly-even.

(2) When $S=\{0,1,5,7\}$ or $\{0,2,6,7\}$, we have

(2.1) The codes $\mathcal C_{[8,m,S]}$ and $\mathcal{C}_{[8,m,S']}$ form a pair of odd-like duadic codes with parameters $[2^m-1, 2^{m-1}, d\geq 2^{\frac{m-1}{2}}+1]$.

(2.2) The dual codes $\mathcal C_{[8,m,S]}^{\perp}$ and $\mathcal{C}_{[8,m,S']}^{\perp}$ form a pair of even-like duadic codes and have parameters $[2^m-1,2^{m-1}-1,d^{\perp}\geq 2^{\frac{m-1}{2}}+2]$.

(2.3) The extended codes $\overline{\mathcal{C}_{[8,m,S]}}$ and $\overline{\mathcal{C}_{[8,m,S']}}$ have parameters $[2^m,2^{m-1},\overline{d}\geq 2^{\frac{m-1}{2}}+4]$. And  they are self-dual and doubly-even.

\end{theorem}
\begin{theorem}
Let $m \equiv{5} \pmod{8}$ be an integer.

(1) When $S=\{0,1,2,6\}$, we have

(1.1) The codes $\mathcal C_{[8,m,S]}$ and $\mathcal{C}_{[8,m,S']}$ form a pair of odd-like duadic codes with parameters $[2^m-1, 2^{m-1}, d\geq 2^{\frac{m-1}{2}}+1]$.

(1.2) The dual codes $\mathcal C_{[8,m,S]}^{\perp}$ and $\mathcal{C}_{[8,m,S']}^{\perp}$ form a pair of even-like duadic codes and have parameters $[2^m-1,2^{m-1}-1,d^{\perp}\geq 2^{\frac{m-1}{2}}+2]$.

(1.3) The extended codes $\overline{\mathcal{C}_{[8,m,S]}}$ and $\overline{\mathcal{C}_{[8,m,S']}}$ have parameters $[2^m,2^{m-1},\overline{d}\geq 2^{\frac{m-1}{2}}+4]$. And  they are self-dual and doubly-even.

(2) When $S=\{0,1,3,7\}$, we have

(2.1) The codes $\mathcal C_{[8,m,S]}$ and $\mathcal{C}_{[8,m,S']}$ form a pair of odd-like duadic codes with parameters $[2^m-1, 2^{m-1}, d\geq 2^{\frac{m-1}{2}}+3]$.

(2.2) The dual codes $\mathcal C_{[8,m,S]}^{\perp}$ and $\mathcal{C}_{[8,m,S']}^{\perp}$ form a pair of even-like duadic codes and have parameters $[2^m-1,2^{m-1}-1,d^{\perp}\geq 2^{\frac{m-1}{2}}+4]$.

(2.3) The extended codes $\overline{\mathcal{C}_{[8,m,S]}}$ and $\overline{\mathcal{C}_{[8,m,S']}}$ have parameters $[2^m,2^{m-1},\overline{d}\geq 2^{\frac{m-1}{2}}+4]$. And  they are self-dual and doubly-even.

(3) When $S=\{0,1,2,7\}$ or $\{0,3,4,6\}$, we have

(3.1) The codes $\mathcal C_{[8,m,S]}$ and $\mathcal{C}_{[8,m,S']}$ form a pair of odd-like duadic codes with parameters $[2^m-1, 2^{m-1}, d]$, where 
$$
d \geq \left\{
\begin{aligned}
&2^{\frac{m-1}{2}}+1, ~~~~&\mbox{if}&~~ m \equiv{5} \pmod{16},\\
&2^{\frac{m-1}{2}}+3, ~~~~~&\mbox{if}&~~ m \equiv{13} \pmod{16}.\\
\end{aligned}
\right.
$$

(3.2) The dual codes $\mathcal C_{[8,m,S]}^{\perp}$ and $\mathcal{C}_{[8,m,S']}^{\perp}$ form a pair of even-like duadic codes and have parameters $[2^m-1,2^{m-1}-1,d^{\perp}]$, where 
$$
d^{\perp} \geq \left\{
\begin{aligned}
&2^{\frac{m-1}{2}}+2, ~~~~&\mbox{if}&~~ m \equiv{5} \pmod{16},\\
&2^{\frac{m-1}{2}}+4, ~~~~~&\mbox{if}&~~ m \equiv{13} \pmod{16}.\\
\end{aligned}
\right.
$$

(3.3) The extended codes $\overline{\mathcal{C}_{[8,m,S]}}$ and $\overline{\mathcal{C}_{[8,m,S']}}$ have parameters $[2^m,2^{m-1},\overline{d}\geq 2^{\frac{m-1}{2}}+4]$. And  they are self-dual and doubly-even.

(4) When $S=\{0,1,3,6\}$ or $\{0,2,4,7\}$, we have

(4.1) The codes $\mathcal C_{[8,m,S]}$ and $\mathcal{C}_{[8,m,S']}$ form a pair of odd-like duadic codes with parameters $[2^m-1, 2^{m-1}, d]$, where 
$$
d \geq \left\{
\begin{aligned}
&2^{\frac{m-1}{2}}+3, ~~~~&\mbox{if}&~~ m \equiv{5} \pmod{16},\\
&2^{\frac{m-1}{2}}+1, ~~~~~&\mbox{if}&~~ m \equiv{13} \pmod{16}.\\
\end{aligned}
\right.
$$

(4.2) The dual codes $\mathcal C_{[8,m,S]}^{\perp}$ and $\mathcal{C}_{[8,m,S']}^{\perp}$ form a pair of even-like duadic codes and have parameters $[2^m-1,2^{m-1}-1,d^{\perp}]$, where 
$$
d^{\perp} \geq \left\{
\begin{aligned}
&2^{\frac{m-1}{2}}+4, ~~~~&\mbox{if}&~~ m \equiv{5} \pmod{16},\\
&2^{\frac{m-1}{2}}+2, ~~~~~&\mbox{if}&~~ m \equiv{13} \pmod{16}.\\
\end{aligned}
\right.
$$

(4.3) The extended codes $\overline{\mathcal{C}_{[8,m,S]}}$ and $\overline{\mathcal{C}_{[8,m,S']}}$ have parameters $[2^m,2^{m-1},\overline{d}\geq 2^{\frac{m-1}{2}}+4]$. And  they are self-dual and doubly-even.

\end{theorem}

\begin{theorem}
Let $m \equiv{7} \pmod{8}$ be an integer

(1) When $S=\{0,1,2,3\}$ or $\{0,1,3,5\}$, we have

(1.1) The codes $\mathcal C_{[8,m,S]}$ and $\mathcal{C}_{[8,m,S']}$ form a pair of odd-like duadic codes with parameters $[2^m-1, 2^{m-1}, d]$, where 
$$
d \geq \left\{
\begin{aligned}
&2^{\frac{m-1}{2}}+1, ~~~~&\mbox{if}&~~ m \equiv{7} \pmod{16},\\
&2^{\frac{m-1}{2}}+3, ~~~~~&\mbox{if}&~~ m \equiv{15} \pmod{16}.\\
\end{aligned}
\right.
$$

(1.2) The dual codes $\mathcal C_{[8,m,S]}^{\perp}$ and $\mathcal{C}_{[8,m,S']}^{\perp}$ form a pair of even-like duadic codes and have parameters $[2^m-1,2^{m-1}-1,d^{\perp}]$, where 
$$
d^{\perp} \geq \left\{
\begin{aligned}
&2^{\frac{m-1}{2}}+2, ~~~~&\mbox{if}&~~ m \equiv{7} \pmod{16},\\
&2^{\frac{m-1}{2}}+4, ~~~~~&\mbox{if}&~~ m \equiv{15} \pmod{16}.\\
\end{aligned}
\right.
$$

(1.3) The extended codes $\overline{\mathcal{C}_{[8,m,S]}}$ and $\overline{\mathcal{C}_{[8,m,S']}}$ have parameters $[2^m,2^{m-1},\overline{d}\geq 2^{\frac{m-1}{2}}+4]$. And  they are self-dual and doubly-even.

(2) When $S=\{0,1,2,4\}$, we have 

(2.1) The codes $\mathcal C_{[8,m,S]}$ and $\mathcal{C}_{[8,m,S']}$ form a pair of odd-like duadic codes with parameters $[2^m-1, 2^{m-1}, d\geq 2^{\frac{m-1}{2}}+3]$.

(2.2) The dual codes $\mathcal C_{[8,m,S]}^{\perp}$ and $\mathcal{C}_{[8,m,S']}^{\perp}$ form a pair of even-like duadic codes and have parameters $[2^m-1,2^{m-1}-1,d^{\perp}\geq 2^{\frac{m-1}{2}}+4]$.

(2.3) The extended codes $\overline{\mathcal{C}_{[8,m,S]}}$ and $\overline{\mathcal{C}_{[8,m,S']}}$ have parameters $[2^m,2^{m-1},\overline{d}\geq 2^{\frac{m-1}{2}}+4]$. And  they are self-dual and doubly-even.

(3) When $S=\{0,4,5,6\}$, we have 

(3.1) The codes $\mathcal C_{[8,m,S]}$ and $\mathcal{C}_{[8,m,S']}$ form a pair of odd-like duadic codes with parameters $[2^m-1, 2^{m-1}, d\geq 2^{\frac{m-1}{2}}+1]$.

(3.2) The dual codes $\mathcal C_{[8,m,S]}^{\perp}$ and $\mathcal{C}_{[8,m,S']}^{\perp}$ form a pair of even-like duadic codes and have parameters $[2^m-1,2^{m-1}-1,d^{\perp}\geq 2^{\frac{m-1}{2}}+2]$.

(3.3) The extended codes $\overline{\mathcal{C}_{[8,m,S]}}$ and $\overline{\mathcal{C}_{[8,m,S']}}$ have parameters $[2^m,2^{m-1},\overline{d}\geq 2^{\frac{m-1}{2}}+4]$. And  they are self-dual and doubly-even.

(4) When $S=\{0,2,3,6\}$ or $\{0,3,5,6\}$, we have

(4.1) The codes $\mathcal C_{[8,m,S]}$ and $\mathcal{C}_{[8,m,S']}$ form a pair of odd-like duadic codes with parameters $[2^m-1, 2^{m-1}, d]$, where 
$$
d \geq \left\{
\begin{aligned}
&2^{\frac{m-1}{2}}+3, ~~~~&\mbox{if}&~~ m \equiv{7} \pmod{16},\\
&2^{\frac{m-1}{2}}+1, ~~~~~&\mbox{if}&~~ m \equiv{15} \pmod{16}.\\
\end{aligned}
\right.
$$

(4.2) The dual codes $\mathcal C_{[8,m,S]}^{\perp}$ and $\mathcal{C}_{[8,m,S']}^{\perp}$ form a pair of even-like duadic codes and have parameters $[2^m-1,2^{m-1}-1,d^{\perp}]$, where 
$$
d^{\perp} \geq \left\{
\begin{aligned}
&2^{\frac{m-1}{2}}+4, ~~~~&\mbox{if}&~~ m \equiv{7} \pmod{16},\\
&2^{\frac{m-1}{2}}+2, ~~~~~&\mbox{if}&~~ m \equiv{15} \pmod{16}.\\
\end{aligned}
\right.
$$

(4.3) The extended codes $\overline{\mathcal{C}_{[8,m,S]}}$ and $\overline{\mathcal{C}_{[8,m,S']}}$ have parameters $[2^m,2^{m-1},\overline{d}\geq 2^{\frac{m-1}{2}}+4]$. And  they are self-dual and doubly-even.

\end{theorem}

\dse{7~~Conclusions}

In this paper, we solve one of open problems proposed by Liu et al. in \cite{12}. Specially, we present several families of binary duadic codes $\mathcal{C}_{[r,m,S]}$ with parameters $[2^m-1,2^{m-1},d]$, where $r$ is even, $m\geq 3$ is odd, $S\subsetneq \mathbb{Z}_{r}$ and $|S|=\frac{r}{2}$. It is shown that the lower bounds of minimum distance $d$ is close to square-root bound. Moreover, the parameters of their dual codes and extended codes are given. We can find that they also have square-root-like lower bounds, and the extended codes are all self-dual and doubly-even.

According to relevant literature, the binary duadic codes with good parameters are as follows:

$\bullet$ Binary quadratic residue codes with parameters $[n, \frac{n+1}{2}, d]$, and their even-weight subcodes have parameters $[n, \frac{n-1}{2}, d+1]$, where $d^2\geq n$ and $n \equiv{\pm 1} \pmod{8}$ is prime. (\cite[\text{Section} 6.6]{5})

$\bullet$ The punctured binary Reed-Muller code PRM$_{2}(\frac{m-1}{2}, m)$ with parameters $[2^m-1, 2^{m-1}, 2^{\frac{m+1}{2}}\\-1 ]$, where $m$ is odd. (\cite{26})

$\bullet$ The binary duadic codes $\mathcal{C}_{(i,m)}$ with defining set $T_{(i,m)}=\{1\leq j \leq n-1: w_{2}(j) \equiv{i} \pmod{2}\}$ (\cite{11}); the codes $\mathcal{C}_{(i_1, i_2, m)}$ with defining set $T_{(i_1, i_2, m)}=\{1\leq j \leq n-1: w_{2}(j) \equiv{i_1 ~\text{or}~ i_2} \pmod{4}\}$ (\cite{12}); the codes $\mathcal{C}_{[6,m,S]}$ with defining set $T_{[6,m, S]}=\{1\leq j \leq n-1: w_{2}(j) \pmod{r} \in S\}$, where $S\subsetneq \mathbb{Z}_{6}$ and $|S|=3$ (\cite{13}). 

Since $2^m-1$ is composite for infinitely many odd $m$, the duadic codes $\mathcal{C}_{[r,m,S]}$ of this paper are not identical with the binary quadratic residue codes. Due to the duadic codes $\mathcal{C}_{[r,m,S]}$ of this paper and the punctured Reed-Muller codes PRM$_{2}(\frac{m-1}{2},m)$ have different defining sets, they are not identity. Through analysis, it can be seen that the binary duadic codes in \cite{11}-\cite{13} are special cases of the codes $\mathcal{C}_{[r,m,S]}$ given in this paper, which can be obtained by taking $r=2$ and $|S|=1$, $r=4$ and $|S|=2$, $r=6$ and $|S|=3$, respectively. Hence, the binary duadic codes $\mathcal{C}_{[r,m,S]}$ for any even $r$ and odd $m$ presented in this paper enrich the theory of binary duadic codes with length $n$ and dimension $\frac{n\pm 1}{2}$. It would be interesting to improve the lower bounds of the binary duadic codes given in this paper.

\dse{Acknowledgments}

This work is supported by the National Natural Science Foundation of China(No. U21A20428, NO. 12171134, NO. 62201009, NO. 12201170) and the Natural Science Foundation of Anhui Province(NO. 2108085QA06, NO. 2108085QA03).

\dse{Statements and Declarations}

\noindent{\bf Declaration of competing interest }

The authors declare that they have no conflicts of interest to this work.\\

\noindent{\bf Data availability }

All data generated or analysed during this study are included in the published article.


\begin{thebibliography}{200}
\frenchspacing
\bibitem{1} Macwilliams, F., Sloane, N.: The theory of error-correcting codes. Elsevier, North Holland (1977)

\bibitem{2} Chien, R.T. : Cyclic decoding procedure for the Bose-Chaudhuri-Hocquenghem codes. 
IEEE Trans. Inf. Theory. 10(4), 357-363 (1964)

\bibitem{3} Forney, G.D.: On decoding BCH codes. IEEE Trans. Inf. Theory. 11(4), 549-557 (1995)

\bibitem{4} Prange, E.: Some cyclic error-correcting codes with simple decoding algorithms Cambridge,
MA, USA, Air Force Cambridge Research Center-TN-58-156 (1958)

\bibitem{5} Huffman, W. C., Pless, V.: Fundamentals Error-Correcting Codes. Cambridge University Press, Cambridge (2003)


\bibitem{6}  Leon, J.S., Masley, J.M., Pless, V.: Duadic codes. IEEE Trans. Inf. Theory 30(5), 709-714 (1984)

\bibitem{7} Leon, J.S.: A probabilistic algorithm for computing minimum weight of large error-correcting codes. IEEE Trans. Inf. Theory 34(5), 1354-1359 (1988) 

\bibitem{8} Pless, V., Masley, J.M., Leon, J.S.: On weights in duadic codes. J. Comb. Theory, Ser. A 44, 6-21 (1987)


\bibitem{9} Ding, C., Lam, K.Y., Xing, C.: Enumeration and construction of all duadic codes of length $p^m$.
Fundam. Inform. 38(1), 149-161 (1999) 


\bibitem{10} Ding, C., Pless, V.: Cyclotomy and duadic codes of prime lengths. IEEE Trans. Inf. Theory 45(2), 453-466 (1999) 

\bibitem{11} Tang C., Ding C.: Binary $[n, \frac{n+1}{2}]$ cyclic codes with good minimum distances. IEEE Trans. Inf. Theory 68(12), 7842-7849 (2022)

\bibitem{12} Liu, H., Li, C., Ding, C.: Five infinite families of binary cyclic codes and their related codes with good parameters. Finite Field Appl. 91, 102270 (2023)

\bibitem{13} Liu, H., Li, C., Qian, H.: Parameters of several families of binary duadic codes and their
related codes. Designs Codes Cryptogr., https://doi.org/10.1007/s10623-023-01285-7.

\bibitem{14} Chen, T., Ding, C., Li, C., Sun, Z.: Four infinite families of ternary cyclic codes with a
square-root-like lower bound. Finite Fields Appl. 92, 102308 (2023)

\bibitem{15}  Shi, M., Tao, S., Kim, J. L., Sol$\acute{e}$, P.: A quaternary analogue of Tang-Ding codes. arXiv:2309.12003v1, 2023.

\bibitem{16} Sun, Z., Li, L., Zhu, S.: A generalization of the Tang-Ding binary cyclic codes. arXiv:2310.20179, 2023.

\bibitem{17} Ding, C.: Cyclotomic constructions of cyclic codes with length being the product of two primes. IEEE Trans. Inf. Theory 58(4), 2231-2236 (2012)

\bibitem{18} Gaborit, P., Nedeloaia, C.S., Wassermann, A.: On the weight enumerators of duadic and quadratic residue codes. IEEE Trans. Inf. Theory 51 (1), 402-407 (2005)

\bibitem{19} Tada, H., Nishimura, S., Hiramatsu, T.: Cyclotomy and its application to duadic codes. Finite Fields Appl. 16(1), 4-13 (2010) 

\bibitem{20}  Xiong, M.: On cyclic codes of composite length and the minimum distance. IEEE Trans. Inf. Theory 64(9), 6305-6314 (2018)


\bibitem{21} Xiong, M., Zhang, A.: On cyclic codes of composite length and the minimum distance II. IEEE Trans. Inf. Theory 67(8), 5097-5103 (2021)

\bibitem{22} Sun, Z.: Several families of binary cyclic codes with good parameters. Finite Fields Appl. 89, 102200 (2023)

\bibitem{23} Sun, Z., Li, C., Ding, C.: An infinite family of binary cyclic codes with best parameters. IEEE Trans. Inf. Theory 70(4), 2411-2418 (2024)

\bibitem{24} Wu, Y., Sun, Z., Ding, C.: Another infinite family of binary cyclic codes with best parameters known. IEEE Trans. Inf. Theory, https://doi.org/10.1109/TIT. 2023. 3310500.

\bibitem{25} Ding, C.: Cyclic codes from cyclotomic sequences of order four. Finite Fields Appl. 23, 8-34 (2013)

\bibitem{26} Assmus, E. F., Key, J. D.: Polynomial codes and finite geometries. In Handbook of Coding Theory, eds. V. S. Pless and W. C. Huffman. Amsterdam: Elsevier, 1269-1343 (1998)

\end{thebibliography}
\end{document}